\documentclass[11pt]{article}
\usepackage{amsfonts}
\usepackage{hyperref}
\usepackage{epsfig}
\usepackage{graphicx}

\newtheorem{theorem}{Theorem}[section]
\newtheorem{definition}[theorem]{Definition}
\newtheorem{proposition}[theorem]{Proposition}
\newtheorem{lemma}[theorem]{Lemma}

\newtheorem{corollary}[theorem]{Corollary}
\newtheorem{remark}[theorem]{Remark}
\newtheorem{notation}[theorem]{Notation}

\renewcommand{\Re}{\mathbb{R}}

\newcommand{\Co}{{H}}

\newenvironment{proof}[1][Proof]{ \noindent \textbf{#1: }}{$\Box$
	\bigskip}

\begin{document}

\title{Error Resilient Space Partitioning\footnote{A preliminary version of the paper~\cite{Our-ICALP} was presented at the ICALP 2021 conference. This work has no associated data.}} 

\author{
	Orr Dunkelman\thanks{Computer Science Department, University of Haifa, Israel. \texttt{orrd@cs.haifa.ac.il}.} \mbox{ and} Zeev Geyzel\thanks{Mobileye, an Intel company -- Jerusalem, Israel. \texttt{zgeyzel@gmail.com}.} \mbox{ and} Chaya Keller\thanks{Department of Computer Science, Ariel University, Ariel, Israel.	\texttt{chayak@ariel.ac.il}. Research supported by the Israel Science Foundation (grant no. 1065/20).} \\
	and Nathan Keller\thanks{Department of Mathematics, Bar Ilan University, Ramat Gan, Israel. \texttt{nkeller@math.biu.ac.il}. Research supported by the European Research Council under the ERC starting grant agreement number 757731 (LightCrypt) and by the BIU Center for Research in Applied Cryptography and Cyber Security in conjunction with the Israel National Cyber Bureau in the Prime Minister's Office.} \mbox{ and} Eyal Ronen\thanks{School of Computer Science, Tel Aviv University. Member of the Check Point Institute for Information Security. \texttt{eyal.ronen@cs.tau.ac.il}.} \mbox{ and} Adi Shamir\thanks{Faculty of Mathematics and Computer Science, Weizmann Institute of Science, Israel. \texttt{adi.shamir@weizmann.ac.il}.} \mbox{   and} Ran J. Tessler\thanks{Incumbent of the Lilian and George Lyttle Career Development Chair, Department of Mathematics, Weizmann Institute of Science, Israel. \texttt{ran.tessler@weizmann.ac.il}. Supported by the Israel Science Foundation (grant no.~335/19) and by a research grant from the center for new scientists of Weizmann Institute of Science.}
}

\maketitle

\begin{abstract}
A major research area in discrete geometry is to consider the best way to partition the $d$-dimensional Euclidean space $\mathbb{R}^d$ under various quality criteria.
In this paper we introduce a new type of space partitioning that is motivated by the problem of rounding noisy measurements from the continuous space $\mathbb{R}^d$ to a discrete subset of representative values. Specifically, we study partitions of $\mathbb{R}^d$ into bounded-size tiles colored by one of $k$ colors, such that   
tiles of the same color have a distance of at least $t$ from each other. Such tilings allow for \emph{error-resilient} rounding, as two points of the same color and distance less than $t$ from each other are guaranteed to belong to the same tile, and thus, to be rounded to the same point.

The main problem we study in this paper is characterizing the achievable tradeoffs between the number of colors $k$ and the distance $t$, for various dimensions $d$. On the qualitative side, we show that in $\mathbb{R}^d$, using $k=d+1$ colors is both sufficient and necessary to achieve $t>0$. On the quantitative side, we achieve numerous upper and lower bounds on $t$ as a function of $k$. In particular, for $d=3,4,8,24$, we obtain sharp asymptotic bounds on $t$, 
as $k \to \infty$. We obtain our results with a variety of techniques including isoperimetric inequalities, the Brunn-Minkowski theorem, sphere packing bounds, Bapat's connector-free lemma, and \v{C}ech cohomology. 

\end{abstract}

\section{Introduction}

Studying various types of partitioning of a continuous space such as $\mathbb{R}^d$ is a central topic in discrete geometry (see, e.g.,~\cite[Chapters~6,12]{dBvKOS08} and the references therein), and each type of partition has different properties and applications within computer science, electrical engineering, and applied mathematics. For example, in {\em error-correcting codes} (that are extensively used in data communication) one tries to squeeze the largest possible number of equal sized disjoint balls into the input space, while in {\em vector quantization}~\cite{Gray84} (that is used extensively in data compression) one tries to completely cover the input space with a small number of tiles whose volumes are as similar as possible. 

In this paper we investigate a new variant of space partitioning, motivated by the problem of rounding noisy measurements.
The basic geometric problem we consider is partitioning $\mathbb{R}^d$ into tiles of volume $\leq 1$, colored by one of $k$ colors, such that tiles of the same color have a distance of at least $t$ from each other.

Such partitions provide a way to handle the inherent discontinuity of rounding processes. For any rounding scheme, there exist points very close to each other that are rounded to different points. The partitioning schemes we consider allow removing this discontinuity using a tiny amount of additional information -- the color of the tile to which the point belongs. As any two points $x_1,x_2$ of the same color such that $\mathrm{dist}(x_1,x_2)<t$ are guaranteed to belong to the same tile, the partition yields a rounding scheme that guarantees that each point is rounded correctly even if an error of $<\frac{t}{2}$ occurred in its measurement, provided that its color is known. Thus, our schemes provide \emph{error-resilient rounding}.

To demonstrate the basic idea, consider the one-dimensional case, in which $x_1$ and $x_2$ are real numbers that have to be rounded to the same nearby integer whenever they are close enough. 
The way we think about the problem is to consider a colored tiling of the real line with two colors:
All the values in $[-0.5, 0.5)$, $[1.5,2.5)$, etc.~are colored by 1, and all the values in $[0.5, 1.5)$, $[2.5,3.5)$, etc.~are colored by 2. The side information provided about $x_1$ is the color of the tile in which it is located, and the way we process $x_2$ is to round it to the center of the closest tile that has the same color as that of $x_1$. The essential property of our partition is that the minimum distance between any two tiles with the same color is $1$, and thus we can ``inflate'' all the tiles of any particular color in order to include any erroneously measured value $x_2$ up to a distance of $0.5$ away from the original tile, and still get non-overlapping tiles that make it possible to uniquely associate such points with original tiles.

\begin{figure}
	\centering
	\includegraphics[width=1.05\columnwidth]{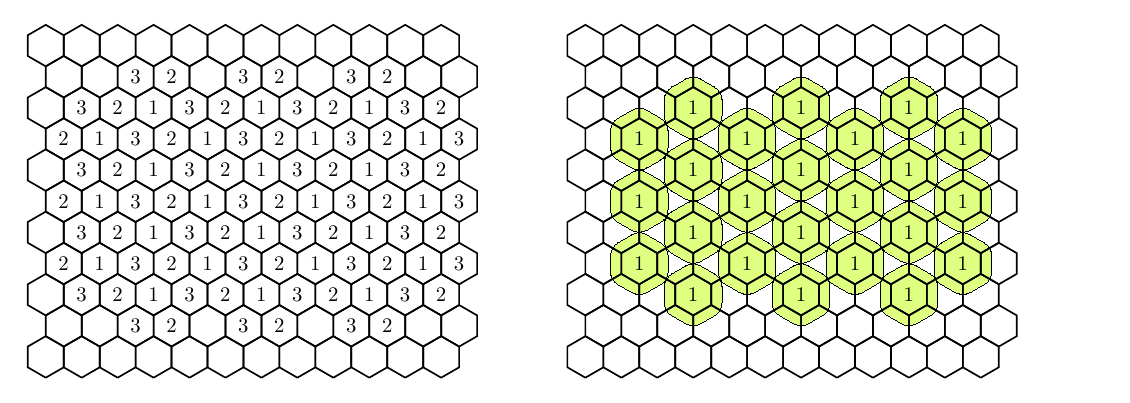}				\caption{A 3-colored hexagonal tiling of the plane, and a maximal non-intersecting inflation of the tiles colored 1}
	\label{fig:hexagons}
\end{figure}

To make this perspective clearer, consider the two-dimensional plane. In the obvious checkerboard tiling by unit squares, we need at least 4 colors if we do not allow equi-colored tiles to touch. We can reduce the number of colors to 3 by considering the hexagonal partition of the plane depicted in the left part of Figure~\ref{fig:hexagons}. Given a two-dimensional point $x_1$, we always round it to the center of the hexagon in which it is located, and given $x_2$ we round it to the center of the nearest hexagon that has $x_1$'s color. To determine the error resilience of this scheme, we inflate all the hexagonal tiles of a particular color by the same amount until they touch each other, as depicted in the right part of Figure~\ref{fig:hexagons}. As it turns out, this natural scheme is not optimal since the inflated hexagons' corners touch prematurely, leaving large gaps between them. A 3-colored tiling with a larger distance between equi-colored tiles (and thus, with a higher error resilience) will be described in Section~\ref{ssec:bricks}. 

An alternative way to view our new type of space partitioning is as a natural generalization of error correcting codes. In standard error correcting codes we consider a discrete collection of points in some high-dimensional space, that are called codewords, and inflate each one of the codewords by the maximal possible amount $\epsilon$ so that all the resultant balls of radius $\epsilon$ around the codewords do not intersect each other. The best codes are obtained when we place the codewords in such a way that this $\epsilon$ is maximized, for some desired density of codewords. In our model, we inflate tiles rather than points by the maximum possible amount $\epsilon$ so that all the inflated tiles do not intersect each other. There are many types of tiles one can consider. We concentrate on tiles that are connnected, of bounded size, and that form a partition of the whole space. Since such tiles already touch each other, we cannot inflate all of them simultaneously; instead, we assign to each tile one of $k$ possible colors, and separately inflate all the tiles of any particular color by the same $\epsilon$. We can thus choose, which space partitioning we start with and how to assign a color to each tile, and our goal is 
to maximize the achievable distance between any two tiles that have the same color.



\subsection{Our results}

In this paper we concentrate on the geometric problem of characterizing the achievable tradeoffs between the number of colors $k$ and the minimum distance between equi-colored tiles $t$, for various dimensions $d$. A discussion on possible applications of the space partitions we study can be found in~\cite{Our-ICALP}.

\paragraph{Qualitative results.} First, we study the basic question, what is the minimal number of colors for which an error resilient rounding scheme is at all possible. We obtain a complete answer in the following theorem.
\begin{theorem}\label{thm:main-qualitative}
    For any $d \in \mathbb{N}$, there exists a partition of $\mathbb{R}^d$ into connected tiles colored by one of $k=d+1$ colors, such that each tile is contained in a box of side length $1$ and there exists $t>0$ for which the distance between any two equi-colored tiles is at least $t$. On the other hand, no such $d$-colored tiling exists. 
\end{theorem}
We prove the existence result by providing an explicit construction, and the non-existence result by an algebraic-topologic argument using either Bapat's connector-free lemma or the Assouad-Nagata dimension theory. We also prove an alternative version of the non-existence result, in which the condition that each tile is contained in a box of side length $1$ is replaced by the condition that each tile is contractible, using the \v{C}ech cohomology theory.

\paragraph{Quantitative results.} On the quantitative side, we study the question of maximizing the guaranteed minimum distance $t$ between equi-colored tiles (or equivalently, maximizing the error resilience) for given $d$ and $k$, assuming (as a normalization) that the maximum volume of a tile is $1$.
Note that here, lower bounds on $t$ correspond to existence results, while upper bounds on $t$ correspond to non-existence results. In the non-existence direction, we obtain several upper bounds on the achievable error resilience, using different techniques from geometry and analysis, including isoperimetry, the Brunn-Minkowski inequality and results on the sphere packing problem. In the existence direction, we construct a variety of concrete tiling schemes with a good error resilience.
In particular, while for $d=2$ and $k=3$ the hexagonal tiling scheme described above yields $t=0.62$, we present a tiling with $t=0.708$, and show that any 3-color tiling satisfies $t \leq 0.826$.
In the general setting of $k=d+1$ colors in $\mathbb{R}^d$, we show that the maximal resilience 
is between $\Omega(\frac{1}{d})$ and $O(\frac{\log d}{\sqrt{d}})$.

In the asymptotic setting where the number $k$ of colors grows to infinity, we obtain tight asymptotic lower and upper bounds on the resilience in dimensions $2,3,8,$ and $24$, as follows.
\begin{theorem}\label{thm:main-asymptotic}
    Let $f(d,k)$ be the maximal value of $t$ such that there exists a partition of $\mathbb{R}^d$ into connected tiles colored by one of $k$ colors, such that the volume of each tile is at most $1$ and the distance between any two equi-colored tiles is at least $t$. Then as $k \to \infty$, we have:
    \[
    \frac{f(2,k)}{k^{1/2}} \to \frac{2^{1/2}}{3^{1/4}} \approx 1.074, \qquad \frac{f(3,k)}{k^{1/3}} \to 2^{1/6} \approx 1.122, \qquad \frac{f(8,k)}{k^{1/8}} \to \sqrt{2}, \qquad \frac{f(24,k)}{k^{1/24}} \to 2.
    \]
\end{theorem}
The upper bound proofs use the  breakthrough results on sphere packing~\cite{CKMRV17,Hales05,V17}, and the matching lower bounds are obtained by explicit constructions. 

The upper and lower bounds on $t$ we obtain in various settings are summarized in Table~\ref{tab:summary}.

\begin{table}[tb]
\scalebox{0.9}{
    \begin{tabular}{|c|c|c|c|c|}
		\hline
		Scenario & Lower Bound & Upper Bound & Techniques & Source \\
		& (LB) on $t$  & (UB) on $t$  & & \\
		\hline \hline
		3 colors & 0.708 & 0.826 & Brunn-Minkowski ineq.~(UB),
		& Sec.~\ref{ssec:BM} (UB), \\
		in $\mathbb{R}^2$ & & & Brick wall tiling~(LB) & Sec.~\ref{ssec:bricks} (LB) \\
		\hline
		4 colors  & 1 & 1.128 & Brunn-Minkowski ineq.~(UB),  & Sec.~\ref{ssec:BM} (UB), \\
		in $\mathbb{R}^2$ & & & Brick wall tiling~(LB) & Sec.~\ref{ssec:bricks} (LB) \\
		\hline
		$k$ colors & $1.074\sqrt{k}-O(1)$ & $1.074\sqrt{k}$ & Circle packing~(UB), & Sec.~\ref{ssec:SP-high} (UB),
		\\
		in $\mathbb{R}^2$ & & & HCR tiling~(LB) &	Sec.~\ref{ssec:hcr} (LB) \\
		\hline
		4 colors & 0.5 & 0.730 & Brunn-Minkowski ineq.~(UB), & Sec.~\ref{ssec:BM} (UB), \\
		in $\mathbb{R}^3$ & & & 3-dim Brick wall~(LB) & Sec.~\ref{ssec:bb} (LB) \\
		\hline
		$k$ colors & $(1.122-o(1))k^{1/3}$ & $1.122k^{1/3}$ & Sphere packing~(UB), & Sec.~\ref{ssec:SP-high} (UB), \\
		in $\mathbb{R}^3$ & & & CPB tiling~(LB) & Sec.~\ref{ssec:CPB} (LB) \\
		\hline
		$k$ colors & $(1.414-o(1))k^{1/8}$ & $1.414k^{1/8}$ & Sphere packing~(UB), & Sec.~\ref{ssec:SP-high} (UB), \\
		in $\mathbb{R}^8$ & & & CPB tiling~(LB) & Sec.~\ref{ssec:CPB} (LB) \\
		\hline
		$k$ colors & $(2-o(1))k^{1/24}$ & $2k^{1/24}$ & Sphere packing~(UB), & Sec.~\ref{ssec:SP-high} (UB), \\
		in $\mathbb{R}^{24}$ & & & CPB tiling~(LB) & Sec.~\ref{ssec:CPB} (LB) \\
		\hline
		$d+1$ colors & $\Omega(\frac{1}{d})$ & $O(\frac{\log d}{\sqrt{d}})$ & Brunn-Minkowski ineq.~(UB), & Sec.~\ref{ssec:BM} (UB), \\
		in $\mathbb{R}^d$ & & & Dimension-reducing tiling~(LB) &  Sec.~\ref{ssec:dim-reduce} (LB) \\
		\hline
		\multicolumn{5}{l}{The `$k$ colors' scenario refers to the asymptotic setting where $k \to \infty$}\\
        \multicolumn{5}{l}{LB -- lower bound, UB -- upper bound,}\\
		\multicolumn{5}{l}{HCR -- honeycomb of rectangles, CPB -- close packing of boxes}
	\end{tabular}
}	\caption{Summary of our lower and upper bounds on $t$, for different values of $d$ and $k$}
	\label{tab:summary}
\end{table}

\subsection{Related work}
\label{sec:sub:related-work}

The new notion of space partitioning we study is closely related to the notion of \emph{sparse partitions}, introduced by Jia et al.~\cite{JiaLNRS05} at STOC'05 (that is, in turn, closely related to the notion of \emph{sparse covers} introduced by Awerbuch and Peleg~\cite{AwerbuchP90a} at FOCS'90). Jia et al.~defined an $(r,\sigma,I)$-partition of a metric space $V$ to be a partition of the space into tiles of diameter $\leq r\sigma$, such that for each $v \in V$, the ball $B_r(v)$ intersects at most $I$ sets in the partition. They proved that for certain parameters $(r,\sigma,I)$, such sparse partitions can be constructed efficiently, and used sparse partitions to devise efficient approximation algorithms for the Universal Steiner tree problem. At FOCS'22, Czumaj et al.~\cite{CzumajJK0Y22} used sparse partitions of $\mathbb{R}^d$ to obtain efficient approximation algorithms for the Euclidean uniform facility location problem. This notion was studied in several other recent papers as well (e.g.,~\cite{BhattacharyaCFJJL25,CzumajJFK0Y22,Filtser24}).

The relation of our notion to sparse partitions comes from the fact that in any partition of $\mathbb{R}^d$ into tiles of $k$ colors such that the distance between any two equi-colored tiles is $\geq t$, any open ball of radius $\frac{t}{2}$ intersects at most $k$ tiles, as it cannot intersect two tiles of the same color. Thus, space partitions of the type we consider yield $(\frac{t}{2},\sigma,d)$ partitions (in the notations of~\cite{JiaLNRS05}), where $\frac{t}{2}\sigma$ is the maximal diameter of a tile. In particular, it is noted in~\cite{CzumajJK0Y22} that a variant of one of the space partition results they obtain can be proved using our \emph{dimension-reducing} tiling (see Section~\ref{ssec:dim-reduce}). However, in general, the two notions are incomparable, as we consider $k$-colored tilings and care only about distances between equi-colored tiles (while in sparse partitions there is a single color and the number of tiles intersecting a given ball is counted), and as we normalize tiles by bounding their volume (while in sparse partitions, the bound is on the diameter).

\subsection{Organization of the paper}

In Section~\ref{sec:qualititave} we study the minimal number of colors that allows achieving $t>0$, and prove the non-existence direction of Theorem~\ref{thm:main-qualitative}. In Section~\ref{sec:upper} we obtain upper bounds on the achievable values of $t$ (in terms of $d$), and in particular, we prove the upper bound part of Theorem~\ref{thm:main-asymptotic}. In Section~\ref{sec:lower} we present several explicit constructions of tilings. In particular, we prove the existence direction of Theorem~\ref{thm:main-qualitative} and the lower bound part of Theorem~\ref{thm:main-asymptotic}. In Appendix~\ref{App:Normalization} we discuss alternative normalizations and distance functions. 

\section{Our Setting}
\label{sec:notation}

In this section we present the basic setting that will be assumed throughout the paper.

\subparagraph*{Colored tiling.} We study \emph{partitions} of $\mathbb{R}^d$ into subsets that are
\emph{connected} (in the topological sense) and \emph{bounded}. Each set in the partition is colored in one of $k$ colors. It will be convenient for us to work with the topological \emph{closures} of the sets of the partition. Those closures form a \emph{tiling} of $\mathbb{R}^d$, where each tile is \emph{connected}, \emph{bounded}, and \emph{closed}, and the tiles intersect only in their boundaries. We note that the boundedness and the connectedness assumptions are natural and were made also in works on sparse partitions (e.g.,~\cite{Filtser24}). 
In some of the results we make additional assumptions on the tiles or drop some of the assumptions; such changes are stated explicitly. 


\subparagraph*{Error resilience and inflation.} In order to compute the error resilience of a given tiling (with respect to the $L_2$ distance), we consider all tiles of the same color and inflate them (i.e., we replace the tile $T$ by the set $T'=\{y: \exists x \in T, \|x-y\|_2<r\}$ for some $r>0$) until they touch each other. Clearly, the error resilience is the maximal $r$, for which such a non-intersecting inflation is possible. We note that in convex geometry, such an inflation $T'$ is called \emph{the outer parallel body of radius $r$} of $T$ (see~\cite[p.~943]{Handbook}) or the open $r$-\emph{neighborhood} of $T$ with respect to the $L_2$-norm. The minimal distance between two equi-colored points in different tiles is denoted by $t$, and so, the error resilience is $t/2$.


\subparagraph*{Normalization.} The $d$-dimensional volume of a figure $T \subset \mathbb{R}^d$ is denoted by $\lambda(T)$. To avoid pathologies, we make the natural assumption that any inflated tile $T'$ satisfies $\lambda(\bar{T}')=\lambda(T')$, where $\bar{T}'$ is the topological closure of $T'$. We normalize the tiling by assuming that \emph{the volume of each tile is bounded by $1$} (like in the 1-dimensional case presented in the introduction, where all tiles are segments of length $1$). All the quantitative bounds we state in Sections~\ref{sec:upper} and~\ref{sec:lower} refer to this assumption.
Normalization with respect to other natural metrics, as well as alternative distance metrics, are discussed in
Appendix~\ref{App:Normalization}.

\section{The Minimal Number of Colors Required for Error Resilience}
\label{sec:qualititave}

In this section we prove the `non-existence' statement of Theorem~\ref{thm:main-qualitative} -- namely, that any colored tiling of $\mathbb{R}^d$ that achieves a positive error resilience, uses at least $d+1$ colors.
We prove two variants of this statement, under different natural assumptions on the tiles. The first assumes that all tiles are uniformly bounded and its proof relies either on Bapat's connector-free lemma (a generalization of Sperner's lemma) or on the \emph{Assouad-Nagata dimension} theory (we provide two proofs, each relying on a different tool). The second statement assumes that the tiles and their non-empty intersections are contractible (while not having to be uniformly bounded) and its proof uses the \v{C}ech cohomology theory. 

The lower bound $d+1$ on the number of required colors is tight; a matching construction for any $d$ is presented in Section~\ref{ssec:dim-reduce}.

\subsection{Lower bound for uniformly bounded tiles}

The main result of this subsection is the following proposition, that proves the non-existence statement of Theorem~\ref{thm:main-qualitative}.

\begin{proposition}\label{Prop:Lower-via-Sperner}
	For any $m>0$, the following holds. Let $T_1,T_2,\ldots$ be a colored tiling of $\mathbb{R}^d$ in $d$ colors, in which each tile is contained in a box with side length $m$. Then there exist two equi-colored tiles whose intersection is non-empty. Consequently, any tiling with minimum distance of at least $t>0$ between equi-colored tiles uses at least $d+1$ colors.
\end{proposition}

We present two proofs of the proposition. The first uses 
\emph{Bapat's connector-free lemma} -- a natural generalization of the classical Sperner's lemma~\cite{Sperner28}, and proves a slightly stronger version of the result (see Proposition~\ref{Prop:Lower-via-Sperner-strong} below). The second proof uses the \emph{Assouad-Nagata dimension} theory. As the tools used in the first proof are more commonly-known, and as it yields a stronger result, we present it in more detail, and then we briefly sketch the second proof.

\subsubsection{Bapat's connector-free lemma -- continuous version}

Since Bapat's lemma was originally proved only in $\mathbb{R}^2$ and in a discrete setting, we first provide a proof of a continuous version in $\mathbb{R}^d$, and then derive Proposition~\ref{Prop:Lower-via-Sperner} from it. Let $\Delta$ be a $d$-simplex in a Euclidean space, i.e., the convex hull of $d+1$ points $x_0,\ldots,x_d$ that do not lie in a $d$-space. The $i$'th face of $\Delta$ is the span of $\{x_j\}_{j\neq i}$. A \emph{connector} in a $d$-simplex is a connected set that intersects all its $(d-1)$-dimensional faces.

Bapat's connector-free lemma asserts the following:
\begin{theorem}\label{thm:bapat_cont}
	Let $C_0,\ldots, C_{d}$ be a cover of a $d$-simplex $\Delta$ by closed sets such that the minimal distance between two connected components of the same set is $\delta>0$. Suppose that the interiors of the sets are disjoint and that no $C_i$ contains a connector. Then $\bigcap_{i=0}^d C_i\neq\emptyset.$
\end{theorem}
In order to prove the theorem we will reduce it to an analogous discrete claim. The reduction is simple, but requires some more terminology.

A \emph{triangulation} $T$ of a simplex $\Delta\subset \mathbb{R}^d$ is a cover of it by
simplices whose interiors are disjoint, such that the intersection of any set of simplices is either empty or the convex hull of some vertices. Note that the vertices of $\Delta$ are, in particular, vertices of the triangulation, and that the faces of $\Delta$ are endowed by an induced triangulation.
We denote by $\partial(T)$ the supremum of distances between vertices that share an edge.
The \emph{$1$-skeleton} of $T$ is the graph formed by the vertices and the edges. A \emph{discrete connector} is a connected subset of the $1$-skeleton of $T$ that contains vertices from each facet of $\Delta.$
We can now state the discrete version of Theorem~\ref{thm:bapat_cont} (which is the actual statement proved by Bapat, for $d=2$).
\begin{theorem}\label{thm:bapat_disc}\cite{Bapat1}
	Suppose that the vertices of $T$ are partitioned into disjoint sets denoted $A_0,A_1,\ldots,A_d$ such that no restriction of the 1-skeleton of $T$ to $A_i$ contains a discrete connector. Then there exists a simplex in $T$ whose $d+1$ vertices belong to different sets $A_i$.
\end{theorem}

\begin{proof}[Proof of Theorem \ref{thm:bapat_cont}, assuming Theorem~\ref{thm:bapat_disc}]
	Assume towards contradiction that there exist sets $C_0,\ldots,C_d$ that cover $\Delta,$ such that no $C_i$ contains a connector, but $\bigcap_{i=0}^d C_i=\emptyset.$
	Define the function $f:C_0\times C_1\cdots\times C_d\to \mathbb{R}_+$ by setting $f(p_0,\ldots,p_d)$ to be the diameter of the convex hull of $(p_0,\ldots,p_d)$, which is the maximal distance between two $p_i$'s. The domain of the function $f$ is compact (here we use the assumption that the minimal distance between two connected components of the same $C_i$ is at least $\delta$) and its range is $\mathbb{R}_{+}$ (since we assumed $\bigcap_{i=0}^d C_i=\emptyset$). Hence, $f$ attains a minimum $\epsilon>0$.
	
	Let
	$\eta = \min(\delta,\epsilon)/3$.
	Let $\tilde{C}_i$ be the $\eta$-thickening of $C_i$ in $\Delta,$ i.e.,
	$\tilde{C}_i=\{p\in \Delta:\mathrm{dist}(p,C_i)<\eta\}$.
	Then $\tilde{C}_i$ is an open cover of $\Delta.$ By the choice of $\eta$, neither $\tilde{C}_i$ contains a connector,\footnote{To be precise, this relies on the slightly stronger assumption that each connected component of $C_i$ is at least $\delta$-far from one of the facets of $\Delta$. While this extra assumption can be avoided in the proof, it clearly holds in our setting so we make it for simplicity.} and $\bigcap_{i\in\{0,\ldots, d\}} \tilde{C}_i=\emptyset$. 
	
	Let $T$ be a triangulation of $\Delta$ such that $\partial(T)<\eta$ and all the vertices of $T$ lie in the interiors of the sets $C_i.$ (Clearly, such a triangulation exists.) Define $A_i$ as the subset of vertices that lie in $\mathrm{int}(C_i).$ These sets are well defined since the interiors of the different $C_i$'s are disjoint. Since $\partial(T)<\eta,$ each edge between two vertices that belong to the same $A_i$ lies in $\tilde{C}_i.$ Indeed, its endpoints are in $C_i$ and any point on the edge is of distance less than $\eta$ to any endpoint, hence it is in $\tilde{C}_i.$
	Thus, since $\tilde{C}_i$ contains no connector, $A_i$ contains no discrete connector.
	
	We can now apply Theorem \ref{thm:bapat_disc} to deduce that there exists a simplex $t= \{v_0,\ldots,v_d\} \in T$ such that $\forall i: v_i \in A_i$. But since $v_i \in C_i$ for all $i$, this implies
	$f(v_0,\ldots,v_d)<\eta<\epsilon,$ a contradiction to the definition of $\epsilon.$ This completes the proof.
\end{proof}

To prove Theorem~\ref{thm:bapat_disc}, we use the classical Sperner's lemma~\cite{Sperner28}. To present it, a few more definitions are due.

 A ($d+1$)-\emph{labelling} of a triangulation $T$ of the simplex $\Delta=\mathrm{conv}(e_0,\ldots,e_d)$ is a function $\ell:V(T) \rightarrow \{0,1,\ldots,d\}$, that is, an assignment of one of $d+1$ colors to each vertex of the triangulation. A $(d+1)$-labelling $\ell$ is called \emph{proper} if $\ell(e_i)=i$, and for each $v \in T$ that belongs to a lower-dimensional face $\mathrm{conv}(e_{i_1},\ldots,e_{i_r})$, we have $\ell(v) \in \{i_1,\ldots,i_r\}$.

\begin{theorem}[Sperner's lemma]
	For any triangulation $T$ of $\Delta$, any proper labelling of $T$ contains a simplex all whose vertices have different labels.
\end{theorem}

\begin{proof}[Proof of Theorem \ref{thm:bapat_disc}]
	We define, using the sets $A_0,\ldots,A_d$, a proper labelling $\ell$ of $T.$ For any $j$ and any $v \in A_j$, $\ell(v)$ is defined as the minimal $i\in\{0,\ldots,d\}$ such that the connected component of $v$ in $A_j$ does not intersect the $i$'th face of $\Delta$. Note that $\ell$ is well-defined, since each $v$ belongs to a single $A_j$ and no $A_j$ contains a connector. Clearly, $\ell(v)\neq i,$ whenever $v$ belongs to the $i$'th face of $\Delta$. Furthermore, this implies that if $v \in \mathrm{conv}(e_{i_1},\ldots,e_{i_r})$, then $\ell(v) \in \{i_1,\ldots,i_r\}$ (as all other colors are forbidden). Hence, $\ell$ is proper.
	
	By Sperner's lemma, applied to the labelling $\ell$,
	there exists a simplex $\{v_0,\ldots,v_d\}\in T$ all whose vertices have different labels. Assume w.l.o.g.~that $\ell(v_i)=i$.
	We want to show that each $v_i$ belongs to a different $A_j,$ which will complete the proof. Assume towards contradiction $v_{i},v_k\in A_j$ for $i\neq k.$ On the one hand, $\ell(v_i)=i\neq k=\ell(v_k).$ On the other hand, $v_i,~v_k$ belong to the same connected component in $A_j,$ hence, by the definition of $\ell$ they must map to the same value. A contradiction, and Theorem \ref{thm:bapat_disc} follows.
\end{proof}

\subsubsection{Proof of Proposition~\ref{Prop:Lower-via-Sperner} via Bapat's Lemma}

We prove the following slightly stronger statement which not only excludes the possibility of $d$-colored tilings with positive minimum distance between equi-colored tiles, but also shows a structural property of all $(d+1)$-colored tilings with positive minimum distance between equi-colored tiles. 

\begin{proposition}\label{Prop:Lower-via-Sperner-strong}
	For any $m>0$, the following holds. Let $T_1,T_2,\ldots$ be a colored tiling of $\mathbb{R}^d$ in $d+1$ colors, in which each tile is contained in a box with side length $m$. If the minimum distance between any two equi-colored tiles is  $\delta>0$, then there exist tiles in all $d+1$ colors that intersect at a point. Consequently, any tiling with positive minimum distance between equi-colored tiles uses at least $d+1$ colors.
\end{proposition}

\begin{proof}
Let $T_1,T_2,\ldots$ be a tiling of $\mathbb{R}^d$ that satisfies the assumptions of the proposition. Consider the restriction of the tiling to a large simplex $\Delta$ (say, of side length $100m$).

For $i=0,\ldots,d$, denote by $C_i \subset \Delta$ the union of all tiles colored $i$, restricted to $\Delta$. Clearly, $C_i$ is a closed set and the distance between any two connected components of $C_i$ is at least $\delta$. 

We claim that no $C_i$ contains a connector. Indeed, a connector cannot include points from different tiles. A single tile is included in a box with side length $m$, and thus, cannot touch all facets of a simplex with side length $100m$. Hence, there is no single-colored connector.

Therefore, we can apply Theorem~\ref{thm:bapat_cont} to deduce that $\bigcap_{i=0}^d C_i \neq \emptyset$, which is exactly the assertion of the proposition.
\end{proof}

\subsubsection{Proof of Proposition~\ref{Prop:Lower-via-Sperner} via the Assouad-Nagata dimension}

An alternative way to obtain a lower bound on the number of colors is using the Assouad-Nagata dimension of metric spaces~\cite{Ass82,Nag58}. In order to present this notion, a few preliminary definitions are needed.

A covering of a metric space $(X,d)$ is called \emph{$C$-bounded} if the diameter of each set in the covering is at most $C$. The \emph{$\delta$-multiplicity} of a covering is the minimal $k \in \mathbb{N}_{\geq 0}$ such that each subset of $X$ with diameter at most $\delta$ has a non-empty intersection with at most $k$ members of the covering.  
\begin{definition}
	The Assouad-Nagata dimension $\mathrm{dim}_N(X)$ of a metric space $(X, d)$ is the smallest integer $n$, for which there exists a constant $C>0$ such that for any $r>0$, the space $X$ admits a $Cr$-bounded covering with $r$-multiplicity at most $n+1$.\footnote{We note that the notion of \emph{sparse partitions} of Jia et al.~\cite{JiaLNRS05}, discussed in Section~\ref{sec:sub:related-work}, is closely related to the Assouad-Nagata dimension. Indeed, the Assouad-Nagata dimension $\mathrm{dim}_N(X)$ is the smallest integer $n$, for which there exists a constant $C>0$ such that for any $r>0$, $X$ admits an $(r,C,n+1)$-partition.}
\end{definition}
A classical result on the Assouad-Nagata dimension (see, e.g.,~\cite[Sec.~2]{LS05}) is:
\begin{proposition}\label{Thm:AN}
	$\mathrm{dim}_N(\mathbb{R}^d)=d$.
\end{proposition}

We claim that this result implies 
Proposition~\ref{Prop:Lower-via-Sperner}.


\medskip

\begin{proof}[Proof of Proposition~\ref{Prop:Lower-via-Sperner}]
	Assume on the contrary that there exists a colored tiling $T_1,T_2,\ldots$ of $\mathbb{R}^d$ with tiles of at most $d$ colors such that each tile is contained in a box of side length $m$ and the distance between any two equi-colored tiles is at least $\delta$. Let $C= m \sqrt{d}/\delta$. We claim that for any $r>0$ there exists a $Cr$-bounded covering of $\mathbb{R}^d$ with $r$-multiplicity at most $d$. Indeed, consider the covering of $\mathbb{R}^d$ with the sets $\frac{r}{\delta}T_1,\frac{r}{\delta}T_2,\ldots$, where $\frac{r}{\delta}T_i =\{\frac{r}{\delta}x: x\in T_i\}$. As each $T_i$ is contained in a box of side length $m$, the covering is $(C r)$-bounded. On the other hand, any $S \subset \mathbb{R}^d$ with diameter at most $r$ intersects at most $d$ sets of the cover, since it cannot intersect sets that come from two equi-colored tiles. Hence, the $r$-multiplicity of the covering is at most $d$. This implies that $\mathrm{dim}_N(\mathbb{R}^d)<d$, which contradicts Proposition~\ref{Thm:AN}. 
\end{proof}

\begin{remark}
We note that while the `Assouad-Nagata dimension' notion is somewhat close to our `resilient coloring' notion, it seems that it cannot be used to obtain the other results presented in the paper. First, this topological notion captures only the question whether there is \emph{some} positive resilience, but not quantitative questions on the maximal possible resilience. Moreover, even with respect to the qualitative question of determining the minimal number of colors required in a resilient colored tiling of $\mathbb{R}^d$, it appears that the Assouad-Nagata dimension provides only a lower bound but not an upper bound. We prove a matching upper bound using a non-trivial construction in Section~\ref{ssec:dim-reduce}.
\end{remark}

\subsection{Lower bound for contractible tiles}
\label{sec:sub:contractible}

Recall that a set in $\mathbb{R}^d$ is called \emph{contractible} if it can be continuously shrunk to a point within the set. (The formal definition is that the identity is homotopic to a constant map.)

Informally, in this section we prove that if the tiles and their non-empty intersections are finite unions of contractible sets (that do not have to be uniformly bounded), then at least $d+1$ colors are required for positive error resilience. 
Due to the possibility of pathologies, the formal statement is a bit more cumbersome:
\begin{proposition}\label{Prop:Contractible}
	Let $T_1,T_2,\ldots$ be a colored tiling of $\mathbb{R}^d$ with positive minimum distance between equi-colored tiles, in which the tiles and all their non-empty intersections are disjoint unions of finitely many closed contractible sets. Assume that the tiling is locally finite (meaning that the number of tiles that intersect any bounded ball $B(0,r)$ is finite) and that all $T_i$'s are bounded (not necessarily uniformly). In addition, assume that each $T_i$ has an open neighborhood $U_i$ such that for any
	$I,$
	\[\bigcap_{i\in I}U_i\neq\emptyset\Leftrightarrow\bigcap_{i\in I}T_i\neq\emptyset,\]
	and the $U_i$'s and their non-empty intersections are disjoint unions of finitely many contractibles.
		Then the number of colors is at least $d+1$.
\end{proposition}

A similar method proves an analogous statement for colored tilings of the sphere $\mathbb{S}^d$ (i.e., the unit sphere in $\mathbb{R}^{d+1}$):
\begin{proposition}\label{Prop:Contractible_Sn}
	Let $T_1,T_2,\ldots,T_N$ be a colored tiling of $\mathbb{S}^d$ with positive minimum distance between equi-colored tiles, in which the tiles and all their non-empty intersections are disjoint unions of finitely many closed contractible sets. Assume that each $T_i$ has an open neighborhood $U_i$ such that for any set of indices $I,$
	\[\bigcap_{i\in I}U_i\neq\emptyset\Leftrightarrow\bigcap_{i\in I}T_i\neq\emptyset,\]
	and the $U_i$'s and their non-empty intersections are disjoint unions of finitely many contractibles.
		Then the number of colors is at least $d+1$.
\end{proposition}

\begin{remark}
	The assumptions on the $U_i$'s are needed for applying to them the \v{C}ech cohomology machinery described below, and the assumptions on the $T_i$'s and their relation to the $U_i$'s are needed for transferring the conclusion of the cohomology argument from the $U_i$'s to the $T_i$'s. We stress that for most natural tilings, the additional assumption on the existence of the neighborhoods $U_i$ follows from the existence of $T_i$'s with the corresponding properties. However, there are topological pathologies, in which this is not the case. 
\end{remark}

\noindent The proof of Propositions~\ref{Prop:Contractible} and~\ref{Prop:Contractible_Sn} uses the notion of \v{C}ech cohomology and classical results regarding its properties. For the ease of reading, we begin with an intuitive explanation of the proof ideas, and then present the formal proof.

\subparagraph*{Intuitive proof.}
The $d$'th (singular) cohomology group is a topological invariant of a manifold that roughly counts ``non trivial holes'' of dimension $d$. A classical result asserts that the $d$'th cohomology group of a $d$-dimensional compact oriented manifold like $\mathbb{S}^d$ is $\mathbb{R}$. (This corresponds to the intuitive understanding that $\mathbb{S}^d$ has one $d$-dimensional hole.) The de-Rham cohomology and the \v{C}ech cohomology are analytic and algebro-geometric/combinatorial invariants, that in many cases agree with their topological cousin. In particular, the $d$'th de-Rham and \v{C}ech cohomologies of $\mathbb{S}^d$ are equal to $\mathbb{R}$ as well.

The $d$'th \v{C}ech cohomology with respect to an open cover of the manifold depends on properties of intersections of $d+1$ sets in that cover. In general, it depends on the sets that form the cover, however, it is known that if these sets and their non-empty intersections are finite disjoint unions of contractibles, then the cohomology groups remain the same, independently of the cover. In particular, if the $d$'th \v{C}ech cohomology with respect to such a cover is non trivial, then there must be $d+1$ sets with a non-empty intersection.

Hence, for our cover $U_1,U_2,\ldots$, we know that its $d$'th \v{C}ech cohomology is $\mathbb{R}$. This readily completes the proof of the proposition for $\mathbb{S}^d$, as this implies that there must be a point that belongs to at least $d+1$ of the $U_i$'s.
The proof in $\mathbb{R}^d$ works in essentially the same way, with cohomology groups replaced by cohomology groups with compact support.

\subparagraph*{Formal proof.}
For the proof we recall the notion of \v{C}ech cohomology with values in the constant sheaf $\underline{\Re}$, and describe the slightly less standard concept of \v{C}ech cohomology with compact support.

\subparagraph*{Definitions.} Let $S$ be either $\Re^{d}$ or a compact manifold such as $\mathbb{S}^{d}.$
Let $\mathcal{U}=\{U_1, U_2,\ldots\}$ be an open cover of $S.$ If $S$ is compact, we assume the collection to be finite. If $S$ is $\Re^{d}$, we assume it to be locally finite and assume in addition that each $U_i$ is bounded.
\begin{itemize}
	\item A \emph{$q-$simplex} $\sigma=(U_{i_0},\ldots, U_{i_q})$ of ${\mathcal {U}}$ is an ordered collection of $q+1$ different sets chosen from $\mathcal {U},$ such that \[\bigcap_{k=0}^{q}U_{i_k}\neq\emptyset.\]
	
	\item For a $q-$simplex $\sigma =(U_{i_k})_{k\in \{0,\ldots ,q\}}$, the $j$'th \emph{partial boundary} is the $(q-1)$-simplex
	\[\partial_{j}\sigma :=(U_{i_k})_{k\in \{0,\ldots ,q\}\setminus \{j\}},\]
	obtained by removing the $j$'th set from $\sigma$.
	
	\item A \emph{$q-$cochain} of $\mathcal{U}$ is a function that associates to any $q-$simplex a real number. The $q-$cochains form a vector space denoted by $C^{q}(\mathcal{U},\underline{\Re}),$ with operations
	\[(\lambda f +\mu g)(\sigma)=\lambda f(\sigma)+\mu g(\sigma), \quad \mbox{where} \quad \lambda,\mu\in\Re,~f,g\in C^{q}(\mathcal{U},\underline{\Re}),~\sigma~\mbox{ is a $q-$simplex}.\]
	Similarly, we define $C^{q}_c(\mathcal{U},\underline{\Re}),$ as the vector space of $q-$cochains \emph{with compact support}, meaning those cochains that assign $0$ to all $q-$simplices, except for finitely many.
	
	\item There is a \emph{differential map} $\delta_{q}:C^{q}(\mathcal{U},\underline{\Re})\to C^{q+1}(\mathcal{U},\underline{\Re})$ whose application to $f\in C^{q}(\mathcal{U},\underline{\Re})$ is the $(q+1)-$cochain $\delta_q(f)$ whose value at a $(q+1)-$simplex $\sigma$ is
	\[(\delta_{q}f)(\sigma)=\sum_{j=0}^{q+1}(-1)^{j}f(\partial_j\sigma).\] The restriction of $\delta_q$ to $C^{q}_c(\mathcal{U},\underline{\Re})$ maps it to $C^{q+1}_c(\mathcal{U},\underline{\Re})$.
	
	\item  It can be easily seen that $\delta_{q+1}\circ\delta_q=0$.
	
	\item The \emph{$q$'th \v{C}ech cohomology group (with compact support)} of $S$ with respect to the cover $\mathcal{U}$ and values in $\underline{\Re}$ is
	\[\check{\Co}^{q}(\mathcal{U},\underline{\Re}):=\mbox{Ker}(\delta_q)/\mbox{Image}(\delta_{q-1}),\]
	\[\check{\Co}^{q}_c(\mathcal{U},\underline{\Re}):=\mbox{Ker}(\delta_q|_{C^{q}_c(\mathcal{U},\underline{\Re})})/\mbox{Image}(\delta_{q-1}|_{C^{q-1}_c(\mathcal{U},\underline{\Re})}).\]
	
	
	\item A cover (by open sets) is \emph{good} if all its sets as well as their multiple intersections are either empty or contractible. It is \emph{almost good} if all non empty intersections are unions of finitely many disjoint contractible components.
\end{itemize}

\subparagraph*{Classical results we use.} The first result we use is the following:
\begin{theorem}\label{theorem:cech_equals_dR}
	If $S$ is a compact smooth orientable manifold (such as $\mathbb{S}^{d}$), and $\mathcal{U}$ is a good or an almost good finite cover, then
	\[\check\Co^{i}(\mathcal{U},\underline{\Re})\simeq \Co^{i}_{dR}(S),\] where the right hand side is the standard de-Rham cohomology group.
	
	 \noindent
	Similarly, if $S=\Re^{d}$ and $\mathcal{U}$ is a locally finite good or almost good cover whose sets are bounded, then
	\[\check\Co^{i}_c(\mathcal{U},\underline{\Re})\simeq \Co^{i}_{dR,c}(S),\] where the right hand side is the $i$'th de-Rham cohomology group with compact support.
\end{theorem}

For further reading about de-Rham cohomology, with or without compact support, we refer the reader to \cite[Sec.~1]{BT}. For further reading about the \v{C}ech cohomology, we refer to \cite[Sec.~8]{BT}. In particular, Theorem \ref{theorem:cech_equals_dR}, for the compact case and good covers is Theorem~8.9 there. The passage to almost good covers is straightforward: In the paragraph that precedes the proof, it is explained that the obstructions to the isomorphism between \v{C}ech and de-Rham cohomologies are given by products of the $i$'th de-Rham cohomology groups, for $i\geq 1,$ of the different intersections $\bigcap_{k=0}^{q} U_{i_k}.$ Since those intersections are disjoint unions of contractibles, their higher cohomology groups vanish, hence there is no obstruction to the isomorphism.

Regarding the case $S=\Re^{d},$ the proof in \cite[Sec.~8]{BT} requires a few small changes: In the statement of Proposition 8.5 there, one needs to replace the de-Rham complex of the manifold with the de-Rham complex with compact support, and the direct product with direct sum.
The maps $r,\delta$ that appear there will still be well defined by our local finiteness assumption on the cover, and the assumption that $U_i$'s are bounded. The proof requires no change.
Then, the double complex in the definition of Proposition~8.8 should also be defined using direct sum rather than direct product, but again there is no change in the proof. Given these changes in definitions, the proof of Theorem~8.9 (also for the almost good case) is unchanged.

\bigskip

The second standard result, which is a consequence of Poincar\'e duality, is the following:
\begin{theorem}\label{theorem:PD}
	For a compact smooth oriented manifold $S$ of dimension $d$ (such as $\mathbb{S}^{d}$),
	\[{\Co}^{d}_{dR}(S)\simeq\Re.\]
	Similarly, for $S=\Re^{d},$ we have $\Co^{d}_{dR,c}(\Re^{d})\simeq \Re$.
\end{theorem}
See, for example, \cite[Sec.~7]{BT} for the compact case, and~\cite[Sec.~4]{BT} for $\Re^{d}.$

\medskip Theorems~\ref{theorem:cech_equals_dR} and~\ref{theorem:PD} yield:
\begin{corollary}\label{cor:non_vanish}
	If $S$ is a compact smooth orientable manifold (such as $\mathbb{S}^{d}$), and $\mathcal{U}$ is a good or an almost good finite cover, then
	\[\check{\Co}^{d}(\mathcal{U},\underline{\Re})=\Re.\]
	Similarly, if $S=\Re^{d}$ and $\mathcal{U}$ is a locally finite good or almost good cover whose sets are bounded, then
	$\check{\Co}^{d}_c(\mathcal{U},\underline{\Re})=\Re$.
\end{corollary}

Now we are ready to prove our assertion.

\subparagraph*{Proof of Propositions~\ref{Prop:Contractible} and~\ref{Prop:Contractible_Sn}.} We show that there must exist $d+1$ $T_i$'s whose intersection is non-empty. This clearly implies that for achieving any positive minimum distance between equi-colored tiles, at least $d+1$ colors are needed.

 Assume on the contrary that any $(d+1)$-intersection of the $T_i$'s is empty. Let $U_i$ be as in the statement of the propositions. Then by definition, they form an almost good cover. All intersections of at least $d+1$ $U_i$'s are empty  by our assumptions. Therefore, there are no $d-$simplices, and so $C^{d}(\mathcal{U},\underline{\Re})=0.$ Thus, in the compact case, $\check{\Co}^{d}(\mathcal{U},\underline{\Re})=0.$ But on the other hand, by Corollary \ref{cor:non_vanish}, \[\check{\Co}^{d}(\mathcal{U},\underline{\Re})\simeq\Re,\] a contradiction. For $\Re^{d}$ the same argument works, with $\check{\Co}^{d}_c(\mathcal{U},\underline{\Re})$ in place of $\check{\Co}^{d}(\mathcal{U},\underline{\Re})$.


\section{On the Maximum Possible Error Resilience}
\label{sec:upper}

In this section we consider tilings of $\mathbb{R}^d$ by tiles $T_1,T_2,\ldots$ of volume at most $1$. Each tile is colored in one of $k \geq d+1$ colors, and our goal is to maximize the minimum distance $t$ between two points of the same color that belong to different tiles. (The maximum is taken over all possible tilings that satisfy the mild regularity conditions stated in Section~\ref{sec:notation} and over all possible colorings.) Clearly, the error resilience of a rounding scheme based on such a colored tiling is $t/2$.

We present three upper bounds on $t$, using the Brunn-Minkowski inequality, results on sphere packing, and the Minkowski-Steiner formula, respectively.
In particular, the upper bounds via sphere packing prove the upper bound statement of Theorem~\ref{thm:main-asymptotic}.

The basic idea behind our upper bound proofs is as follows. Assume we have a colored tiling of $\mathbb{R}^d$, with minimum distance $t$. Pick a single color -- say, black -- and consider all black tiles inside a large cube $S$.
We obtain a new collection of tiles $T'_1,T'_2,\ldots,T'_m$ that covers part of $S$. The assumption that the minimum distance between two equi-colored points in different tiles is $t$ implies that if we inflate each black tile $T'_i$ into its open parallel outer body of radius $t/2$,
\begin{equation}\label{Eq:Blowup1}
T''_i =\{x: \exists y \in T'_i, \|x-y\|_2<t/2\},
\end{equation}
then the open parallel bodies $T''_i$ are pairwise disjoint. Hence, the sum of their volumes essentially cannot exceed the volume of the large cube, and this allows bounding $t$ from above.

\subsection{An upper bound using the Brunn-Minkowski inequality}
\label{ssec:BM}

The open parallel bodies $T''_i$ can be represented in terms of the \emph{Minkowski sum} of sets in $\mathbb{R}^d$.
\begin{definition}
	For $A,B \subset \mathbb{R}^d$, the Minkowski sum of $A,B$ is $A+B=\{a+b:a \in A, b \in B\}$.
\end{definition}
In terms of this definition, we have
\begin{equation}\label{Eq:BM1}
T''_i = T'_i + B(0,t/2),
\end{equation}
where $B(0,t/2)$ is an open ball of radius $t/2$ around the origin. This allows us to lower bound the volume of each $T''_i$, using the classical Brunn-Minkowski (BM) inequality (see, e.g.,~\cite{Busemann58}). Recall the inequality asserts the following.
\begin{theorem}[Brunn-Minkowski]
	Let $A,B$ be compact sets in $\mathbb{R}^d$. Then
	\[
	\lambda(A+B)^{1/d} \geq \lambda(A)^{1/d}+\lambda(B)^{1/d},
	\]
	where $\lambda(X)$ is the volume of $X$ (formally, the $d$-dimensional Lebesgue measure of $X$).
	
\end{theorem}

\begin{proposition}\label{Prop:BM}
	Let $T_1,T_2,\ldots$ be a $k$-colored tiling of $\mathbb{R}^d$, with tiles of volume $\leq 1$ and minimal distance $t$. Then
	\[
	t \leq \left(\frac{2\Gamma(d/2+1)^{1/d}}{\sqrt{\pi}} \right) \cdot (k^{1/d}-1),
	\]
	where $\Gamma(\cdot)$ is the Gamma function.
\end{proposition}

\begin{proof}
	Consider a cube $S$ such that $\lambda(S)=n^d$ (for some `large' $n$). By the pigeonhole principle, there exists a color (say, black) that covers at least $\frac{n^d}{k}$ of the volume of $S$. Look at the black tiles whose intersection with $S$ is non-empty, and denote their intersections with $S$ by $T'_1,T'_2,\ldots,T'_m$. Hence, we have $m$ `black' subsets of $S$, each of volume at most 1, whose total volume is at least $\frac{n^d}{k}$.
	
	For each $T'_i$, define $T''_i = T'_i + B(0,t/2)$. By assumption, the regions $T''_i$ are disjoint. Furthermore, they are included in $S+B(0,t/2)$ whose volume is less than $(n+t)^d$. Hence,
	\begin{equation}\label{Eq:BM2}
	\sum_i \lambda(T''_i) \leq (n+t)^d.
	\end{equation}
	By the Brunn-Minkowski inequality, we have
	\[
	\forall i: \lambda(T''_i)^{1/d} \geq \lambda(T'_i)^{1/d} + (b_{t/2})^{1/d},
	\]
	where $b_{t/2}$ is the volume of the $d$-dimensional ball $B(0,t/2)$. Thus,
	\[
	\forall i: \lambda(T''_i) \geq \sum_{j=0}^d {{d}\choose{j}} \lambda(T'_i)^{j/d} (b_{t/2})^{1-\frac{j}{d}}.	
	\]
	Summing over $i$ and using~(\ref{Eq:BM2}), we get
	\begin{equation}\label{Eq:BM3}
	(n+t)^d \geq \sum_{i=1}^m \sum_{j=0}^d {{d}\choose{j}} \lambda(T'_i)^{j/d} (b_{t/2})^{1-\frac{j}{d}}.	
	\end{equation}
	As $0 \leq \lambda(T'_i) \leq 1$, for any $0 \leq j \leq d$ we have $\sum_i \lambda(T'_i)^{j/d} \geq \sum_i \lambda(T'_i) \geq \frac{n^d}{k}$, and hence we obtain
	\[
	(n+t)^d \geq \frac{n^d}{k} \cdot \left(1+b_{t/2}^{1/d} \right)^d.
	\]
	This implies
	\[
	\left(1+\frac{t}{n} \right)k^{1/d}-1 \geq b_{t/2}^{1/d} = \frac{\pi^{1/2}}{\Gamma(\frac{d}{2}+1)^{1/d}} \cdot \frac{t}{2}.
	\]
	Letting $n \rightarrow \infty$ and rearranging, we obtain
	\[
	t \leq \left(\frac{2\Gamma(\frac{d}{2}+1)^{1/d}}{\sqrt{\pi}} \right) \cdot (k^{1/d}-1),
	\]
	as asserted.
\end{proof}

\subparagraph*{Asymptotic upper bound.} Note that we have
\[ 
\Gamma \left(\frac{d}{2}+1 \right)^{1/d} = \left(\frac{1}{\sqrt{2e}}+o_d(1)\right)\sqrt{d},
\]
where $o_d(1)$ denotes a function of $d$ which tends to $0$ as $d \to \infty$. Hence, for large number $k \gg d$ of colors, Proposition~\ref{Prop:BM} gives the upper bound 
\[
t \leq \left(\sqrt{\frac{2}{\pi e}}+o_d(1) \right) \sqrt{d} \cdot k^{1/d}.
\]
This bound is not far from being tight. Indeed, its dependence on $k$ is correct, as it can be easily matched by a periodic cubic tiling, in which each tile is a cube with side length 1 and the basic unit is a large cube with side length $k^{1/d}$ that contains each color in exactly one tile (in the same order). Moreover, even regarding the `coefficient' of $k^{1/d}$, the optimal asymptotic upper bounds for $d=3,8,24$ that we obtain in Section~\ref{ssec:SP-high} via the sphere packing problem, improve over this bound by only a small factor.

\subparagraph*{Upper bound for $k=d+1$ colors.} Note that
$
(d+1)^{1/d}-1 = (1+o_d(1))\frac{\ln(d)}{d}$. 
Therefore, the bound we obtain in this case is
\[
t \leq \left(\sqrt{\frac{2}{\pi e}}+o_d(1) \right)\frac{\ln d}{\sqrt{d}},
\]
which implies that the error resilience decreases to zero as $d$ tends to infinity. For comparison, the lower bound we obtain in Section~\ref{ssec:dim-reduce} is $t \geq \Omega(\frac{1}{d})$.

\subparagraph*{Upper bounds for small values of $d,k$.}  For $d=3,k=4$, the bound is
\[
t \leq \frac{2\Gamma(2.5)^{1/3}}{\sqrt{\pi}} \cdot (4^{1/3}-1) \approx 0.729.
\]
For $d=2$ and $k=3,4$, the upper bounds we obtain are $t \leq 0.826$ and $t \leq 1.128$, respectively. For comparison, the best constructions we have in these settings are $t=0.5$ for $d=3,k=4$, $t=\frac{1}{\sqrt{2}}$ for $d=2,k=3$, and $t=1$ for $d=2,k=4$ (see Section~\ref{sec:sub:bricks}).

\subparagraph*{Discussion.}
The upper bound given by Proposition~\ref{Prop:BM} is loose in two ways. One source of loss is the application of the Brunn-Minkowski inequality. Here, the inequality is tight if the tiles are \emph{balls}, and the farther they are from balls, the larger is the loss. Another source of loss is the space left between the inflations, that is not taken into account in the proof.

Interestingly, there is a dichotomy between these two sources of loss. As follows from the sphere packing problem, when the tiles are balls (and so, there is no loss in the BM inequality), the space between the inflations (and so, the loss of the second type) is relatively large. The space between the inflations can be made smaller if the tiles are taken to be polytopes with a few vertices. However, this comes at the expense of increased loss in the BM inequality, as is demonstrated in Section~\ref{ssec:Steiner}.

\subparagraph*{Optimality of our 1-dimensional tiling.} The argument described above gives an easy proof of the optimality of the 1-dimensional tiling presented in the introduction. Indeed, consider a 2-colored tiling of the line and look at the segment $I = [-n,n]$ for some large $n$. By the pigeonhole principle, we may assume that black tiles cover at least half of $I$. By the 1-dimensional Brunn-Minkowski inequality, for each black tile $T'_i \subset I$ and the corresponding outer parallel body $T''_i=T'_i + (-t/2,t/2)$, we have $\lambda(T''_i) \geq \lambda(T'_i)+t$. As $\forall i: \lambda(T'_i) \leq 1$, there are at least $n$ tiles. Since the $T''_i$'s are pairwise disjoint and included in $[-n-1,n+1]$, we obtain
\[
2n+2 \geq \sum_i \lambda(T''_i) \geq \sum_i \lambda(T'_i)+ \sum_i t \geq n + nt,
\]
and thus, $t \leq (n+2)/n$. By letting $n$ tend to infinity, we obtain $t \leq 1$, implying that the tiling presented in the introduction, in which the black tiles are all segments of the form $[2n-0.5,2n+0.5)$, obtains the maximum possible value of $t$. 

\subparagraph*{Related work.} We note that in the setting of sparse partitions, the Brunn-Minkowski inequality was used to obtain upper bound results by Filtser~\cite{Filtser24}. 


\subsection{An upper bound using the Sphere Packing problem}
\label{ssec:SP-high}

Our second upper bound uses reduction to the classical \emph{sphere packing} problem, that asks for the maximal possible \emph{density} of a set of non-intersecting congruent spheres in $\mathbb{R}^d$.
\begin{definition}
	The density of a sphere packing (i.e., collection of pairwise disjoint congruent spheres) $P= \cup P_i$ in $\mathbb{R}^d$ is
	\[	
	\limsup_{r \rightarrow \infty}  \frac{\lambda \left(B(0,r) \cap \bigcup P \right)}{\lambda(B(0,r))}.
	\]
\end{definition}
Intuitively, this measures the fraction of the volume of a large ball covered by the packing.
\begin{notation}
	Denote the maximal density of a sphere packing in $\mathbb{R}^d$ by $\delta_d$, and the volume of the unit ball $B(0,1) \subset \mathbb{R}^d$ by $$v_d = \frac{\pi^{\frac{d}{2}}}{\Gamma(\frac{d}{2}+1)}.$$
\end{notation}

\begin{proposition}\label{Prop:SP-higher}
	Let $T_1,T_2,\ldots$ be a tiling of $\mathbb{R}^d$ in $k$ colors, with tiles of volume $\leq 1$ and minimum distance $t$. Then
	\[
	t \leq \left(2 \left(\frac{\delta_d}{v_d}\right)^{1/d} \right) \cdot k^{1/d} = \left(\frac{2\Gamma(\frac{d}{2}+1)^{1/d} \cdot \delta_d^{1/d}}{\sqrt{\pi}} \right) \cdot k^{1/d}.
	\]
\end{proposition}	
Note that the asymptotic upper bound of Proposition~\ref{Prop:SP-higher} is stronger than the asymptotic upper bound that follows from Proposition~\ref{Prop:BM} by the constant factor $(\delta_d)^{1/d}$. For small values of $k$, the upper bound given by Proposition~\ref{Prop:BM} is stronger.

\begin{proof}
	Let $T$ be a $k$-colored tiling of $\mathbb{R}^d$ that satisfies the assumptions of the proposition, and consider the sequence of balls  $\{B(0,n)\}_{n=1,2,3,\ldots}$. By the pigeonhole principle, there exists a color (say, black) such that for each $n_{\ell}$ in an infinite subsequence $\{n_\ell\}_{\ell=1,2,\ldots}$, the intersection of the black tiles with the ball $B(0,n_\ell)$ has volume of at least
	\[
	\frac{\lambda(B(0,n_\ell))}{k} = \frac{n_\ell^d \cdot v_d}{ k}.
	\]
	As the volume of each tile is at most $1$, we know that for each $n_\ell$, the number of black tiles that intersect $B(0,n_\ell)$ is at least $\frac{n_\ell^d \cdot v_d}{k}$.
	
	 \noindent Pick some value $n_{\ell}$, denote the intersections of black tiles with $B(0,n_{\ell})$ by $T'_1,T'_2,\ldots$, and take one point $x_i$ from each tile $T'_i$. As the minimal distance between two black points in different tiles is $t$, balls of radius $t/2$ around the points $x_i$ are pairwise disjoint. Hence, their total volume is at least
	\[
	\frac{n_\ell^d \cdot v_d}{k} \cdot \left(\frac{t}{2}\right)^d \cdot v_d.
	\]
	On the other hand, each such ball is contained in the ball $B(0,(n_{\ell}+t))$ (since its radius is $t/2$, and it contains a point in $B(0,n_\ell)$). This implies that for any $\epsilon>0$ and for a sufficiently large $\ell=\ell(\epsilon)$, the total volume of these balls must be smaller than $(1+\epsilon)\delta_d \cdot \lambda(B(0,n_{\ell}+t))$, as otherwise, the infinite collection of the balls $B(x_i,t/2)$ (where for each ball $B(0,n_\ell)$ we select $x_i$'s in the way described above, respecting the $x_i$'s selected for smaller values of $n_\ell$) would be a sphere packing of $\mathbb{R}^d$ whose density is larger than $\delta_d$.	Therefore, for a sufficiently large $n_\ell$, we have
	\[
	\frac{n_\ell^d \cdot v_d}{k} \cdot \left(\frac{t}{2} \right)^d \cdot v_d \leq (1+\epsilon)\left(1+\frac{t}{n_{\ell}}\right)^d\delta_d \cdot n_\ell^d v_d,
	\]
	and letting $\epsilon \rightarrow 0$ and $n_{\ell} \rightarrow \infty$, we obtain $
	t \leq 2 (\delta_d/v_d)^{1/d} \cdot k^{1/d}$,
	as asserted.
\end{proof}


\subparagraph*{Discussion.} In the two last decades, there has been a tremendous progress in the research of the sphere packing problem. In 2005, Hales (\cite{Hales05}, see also~\cite{Hales17}) solved the problem for $d=3$, proving a 17'th century conjecture of Kepler. In 2017, in a beautiful short paper, Viazovska~\cite{V17} solved the problem for $d=8$, and shortly after, Cohn, Kumar, Miller, Radchenko, and Viazovska~\cite{CKMRV17} used Viazovska's method along with other tools to solve the problem for $d=24$. For other dimensions, the problem is still open. We can use the results of~\cite{CKMRV17,Hales05,V17}, along with the value of $\delta_2$ that was obtained already by Lagrange, to obtain tight asymptotic upper bounds on $t$ in dimensions $2,3,8,$ and $24$, which are the upper bounds stated in Theorem~\ref{thm:main-asymptotic}.
\begin{itemize}
	\item For $d=2$, Lagrange (1773) showed that $\delta_2=\frac{\pi}{2\sqrt{3}} \approx 0.907$. Hence, we obtain the bound $\frac{t}{k^{1/2}}\leq \frac{2^{1/2}}{3^{1/4}} \approx 1.074$.
	
	\item For $d=3$, Hales~\cite{Hales05,Hales17} showed that $\delta_3=\frac{\pi}{3\sqrt{2}} \approx 0.740$. Hence, we obtain the bound $\frac{t}{k^{1/3}}\leq 2^{1/6} \approx 1.122$.
	
	\item For $d=8$, Viazovska~\cite{V17} showed that $\delta_8=\frac{\pi^4}{2^4 4!} \approx 0.254$. Hence, we obtain the bound $\frac{t}{k^{1/8}}\leq \sqrt{2} \approx 1.414$.
	
	\item For $d=24$, Cohn et al.~\cite{CKMRV17} showed that $\delta_{24}=\frac{\pi^{12}}{12!} \approx 0.0019$. Hence, we obtain the bound $\frac{t}{k^{1/24}} \leq 2$.
\end{itemize}
As we show in Section~\ref{sec:lower}, all these bounds are asymptotically tight.

Using the same method, we can leverage any upper bound for the sphere packing problem (namely, upper bound on $\delta_d$) into an upper bound on the error resilience of a rounding scheme in the corresponding dimension. The best currently known bound on $\delta_d$ for large $d$ is by Sardari and Zargar~\cite{SardariZ24}, who obtained a constant-factor improvement over the classical Kabatiansky-Levenshtein~\cite{KL78} bound $\delta_d \leq 2^{-(0.5990+o(1))d}$ (see also~\cite{CZ14}). A list of conjectured bounds for $d \leq 10$ can be found in~\cite{CS95}.


\subsection{An upper bound using the Minkowski-Steiner formula} \label{ssec:Steiner}

This upper bound is applicable in the plane, and under the additional assumption that the tiles $T_i$ are convex. It uses another well-known result, called the \emph{Minkowski-Steiner formula}.
\begin{theorem}[Minkowski-Steiner]\label{Thm:MS}
	Let $A$ be a convex compact set in $\mathbb{R}^2$. Then for any circle $B(0,r)$,
	\begin{equation}\label{Eq:MS1}
	\lambda(A+B(0,r)) = \lambda(A)+\ell(\partial A)r + \pi r^2,
	\end{equation}
	where $\ell(\partial(A))$ is the (1-dimensional) length of the boundary of $A$.
\end{theorem}

\medskip \noindent By substituting~(\ref{Eq:MS1}) into the proof of Proposition~\ref{Prop:BM} (for $d=2$), we obtain
\begin{equation}\label{Eq:MS2}
(n+t)^2 \geq \sum_{i=1}^m \lambda(T''_i) = \sum_{i=1}^m \lambda(T'_i) + \frac{t}{2} \cdot \sum_{i=1}^m \ell(\partial(T'_i)) + \frac{m \pi t^2}{4}.
\end{equation}
To proceed, we may bound the length of the boundary of each $T'_i$ in terms of the area $\lambda(T'_i)$. If we do not make further assumptions on the $T'_i$'s, then the best possible bound of this type is the \emph{isoperimetric inequality}, which asserts that
\[
\ell(\partial(A)) \geq \sqrt{4\pi}\sqrt{\lambda(A)},
\]
for any region $A$ in the plane bounded by a closed curve. Plugging this into~(\ref{Eq:MS2}) yields exactly the same upper bound as Proposition~\ref{Prop:BM} (which comes by no surprise, as the tightness case of the isoperimetric inequality in the plane is circles, just like the tightness case in the proof of Proposition~\ref{Prop:BM}).

However, if we further assume that each tile is a convex polygon, then we can obtain an improved bound, as a function of the number of vertices in each such polygon. We use another classical result, going back to an Ancient Greek mathematician:
\begin{theorem}[Zenodorus]\label{thm:Zenodorus}
	Among all polygons on $n$ vertices with the same area, the perimeter is minimized for the regular $n$-gon.
\end{theorem}
Recall that the perimeter of a regular $l$-gon $P_l$ of area $\lambda(P_l)$ is
\[
\ell(\partial(P_l)) = \frac{2\sqrt{l}}{\sqrt{\cot(\pi/l)}} \cdot \sqrt{\lambda(P_l)}.
\]
Plugging this into~(\ref{Eq:MS2}), we obtain the following.
\begin{proposition}\label{Prop:MS}
	Let $T_1,T_2,\ldots$ be a tiling of the plane in $k$ colors, with tiles of area $\leq 1$ and minimal distance $t$. Assume than all tiles are convex polygons with at most $l$ vertices. Then
	\[
	t \leq \frac{-\alpha_l+\sqrt{\alpha_l^2-4(1-k)\pi}}{\pi},
	\qquad \mbox{where} \qquad \alpha_l = \frac{2\sqrt{l}}{\sqrt{\cot(\pi/l)}}.
	\]
\end{proposition}

\begin{proof}
	First, we follow the proof sketch presented above until Equation~(\ref{Eq:MS2}), that asserts
	\[
	(n+t)^2 \geq \sum_{i=1}^m \lambda(T''_i) = \sum_{i=1}^m \lambda(T'_i) + \frac{t}{2} \cdot \sum_{i=1}^m \ell(\partial(T'_i)) + \frac{m \pi t^2}{4}.
	\]
	Using Theorem~\ref{thm:Zenodorus}, we obtain
	\[
	(n+t)^2 \geq \sum_{i=1}^m \lambda(T'_i) + \frac{t}{2} \cdot \alpha_l \cdot \sum_{i=1}^m \sqrt{\lambda(T'_i)} + \frac{m \pi t^2}{4},
	\]
	where $\alpha_l=\frac{2\sqrt{l}}{\sqrt{\cot(\pi/l)}}$. As $m\geq \frac{n^2}{k}$ and $\sum_i \sqrt{\lambda(T'_i)} \geq \sum_i \lambda(T'_i) \geq \frac{n^2}{k}$, this implies
	\[
	(n+t)^2 \geq \frac{n^2}{k} \cdot \left(1 + \frac{t}{2} \cdot \alpha_l+\frac{\pi t^2}{4}\right),
	\]
	and consequently, $(1+\frac{t}{n})^2 \cdot k \geq 1+\frac{\alpha_l}{2} t + \frac{\pi}{4}t^2$. Letting $n \rightarrow \infty$, we get $k \geq 1+\frac{\alpha_l}{2} t + \frac{\pi}{4}t^2$. Solving the quadratic inequality, we obtain  	
	\[
	t \leq \frac{-\alpha_l+\sqrt{\alpha_l^2-4(1-k)\pi}}{\pi},
	\]
	as asserted.
\end{proof}

\subparagraph*{Discussion.} For small values of $l$, the upper bound obtained in Proposition~\ref{Prop:MS} is rather strong. For example, for $l=3$ we obtain $t \leq 0.707$ for 3 colors and $t \leq 0.985$ for four colors. 
In Section~\ref{sec:sub:bricks} we will present tilings using rectangular tiles that achieve $t = \frac{1}{\sqrt{2}}$ for $3$ colors and $t=1$ for $4$ colors.   Proposition~\ref{Prop:MS} shows that no tiling with triangular tiles can achieve these bounds.

As $l$ increases, the bound of Proposition~\ref{Prop:MS} becomes weaker and approaches the bound of Proposition~\ref{Prop:BM}, since circles (for which Proposition~\ref{Prop:BM} is tight) can be approximated to any precision by convex polygons with a sufficiently large number of vertices.

\subparagraph*{Higher dimensions.}
The Minkowski-Steiner formula has a higher-dimensional analogue:
\begin{equation}\label{Eq:MS-higher}
\lambda(A+B(0,r)) = \lambda(A)+\lambda_{d-1}(\partial A)r + \sum_{j=2}^{d-1} \lambda_j(A) r^j + \frac{2\pi^{d/2}}{d\Gamma(d/2)} r^d,
\end{equation}
where $\lambda_{d-1}(\partial(A))$ is the $(d-1)$-dimensional volume of the boundary of $A$, and $\lambda_j(A)$ are continuous functions of $A$ called \emph{mixed volumes} (which depend on volumes of certain \emph{outer parallel bodies}, and were the original context in which outer parallel bodies were defined).

In order to obtain an effective upper bound on the error resilience using~(\ref{Eq:MS-higher}), one has to effectively bound the mixed volumes $\lambda_j(A)$ from below in terms of $\lambda(A)$, which seems to be a challenging task.

\section{Constructions of Error Resilient Space Partitions}
\label{sec:lower}

In this section (like in Section~\ref{sec:upper}), we consider tilings of $\mathbb{R}^d$ by tiles $T_1,T_2,\ldots$ of volume at most $1$. Each tile is colored in one of $k \geq d+1$ colors, and our goal is to maximize the minimum distance $t$ between two points of the same color that belong to different tiles.

We obtain lower bounds on the maximum achievable value of $t$ for various values of $d,k$, by constructing explicit tilings. First, we present \emph{brick-wall} tilings, that provide lower bounds for small values of $d,k$. Then we present tilings based on \emph{close sphere packing}, that show the tightness of our asymptotic bounds for $d=2,3,8,24$ and $k \to \infty$, thus proving the `lower bound' part of Theorem~\ref{thm:main-asymptotic}. Finally, we construct a \emph{dimension-reducing} tiling, which shows that for any dimension $d$, positive error resilience can be obtained with $d+1$ colors, thus proving the existence part of Theorem~\ref{thm:main-qualitative}.

\subsection{Brick-wall tilings}
\label{sec:sub:bricks}

We begin with a tiling of the plane, and then use it to construct a tiling of $\mathbb{R}^3$.


\subsubsection{2-dimensional brick wall}
\label{ssec:bricks}

In the 2-dimensional brick wall tiling with $k$ colors, demonstrated in Figure~\ref{fig:brick_wall}, each tile is a rectangle with side lengths $\sqrt{\frac{k-2}{2}}$ and $\sqrt{\frac{2}{k-2}}$ (and so, the area of each tile is $1$). The tiling is periodic, where the basic unit is two rows of adjacent rectangles, colored in a round robin fashion. For an even $k$, the second row is placed exactly below the first row, and the sequence of colors is shifted by $\frac{k}{2}$. For an odd $k$, the second row is indented by half a brick (making the tiling look like a brick wall), and the sequence is shifted by $\frac{k+1}{2}$.

\begin{figure}
	\centering
	\includegraphics[width=1.0\columnwidth]{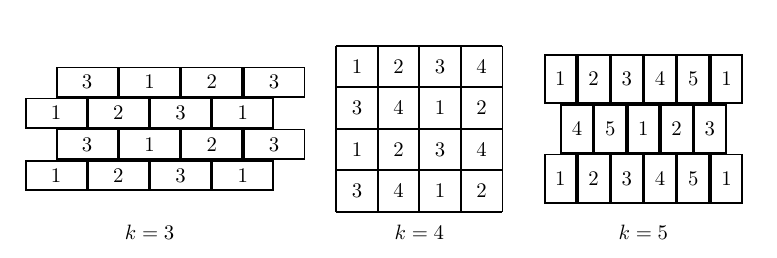}				\caption{The 2-dimensional brick wall tiling for $k=3,4,5$ colors. The ratio between the width and the height of each tile is $2:(k-2)$.}
	\label{fig:brick_wall}
\end{figure}

It is easy to see that the minimal distance between two equi-colored points in different tiles is $\sqrt{\frac{k-2}{2}}$. (This distance is attained both in the vertical and in the horizontal directions. Having the same minimal distance in both directions is the optimization that dictates the side lengths of the bricks.) In particular, we obtain the lower bounds $t \geq \frac{1}{\sqrt{2}}$ for 3 colors, $t \geq 1$ for 4 colors, and $t \geq \sqrt{\frac{3}{2}}$ for 5 colors.


\subsubsection{3-dimensional brick wall}
\label{ssec:bb}


The 3-dimensional brick wall (3BW) tiling, demonstrated in Figure~\ref{fig:3BW}, is a periodic tiling of $\mathbb{R}^3$, colored in 4 colors. In order to present the tiling, we need an auxiliary notation.

\begin{figure}
	\centering
	\includegraphics[width=0.9\columnwidth]{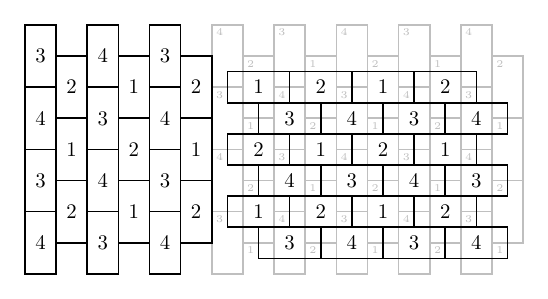}				\caption{The 3-dimensional brick wall tiling}
	\label{fig:3BW}
\end{figure}

\subparagraph*{Notation.} Consider the slab $D=\mathbb{R}\times \mathbb{R} \times [z_1,z_2] \subset \mathbb{R}^3$. We say that a tiling $T_1,T_2,\ldots$ of $D$ is a \emph{fattened plane tiling} if there exists a tiling $T'_1,T'_2,\ldots$ of the plane such that $\forall i: T_i = T'_i \times [z_1,z_2]$.

\subparagraph*{The structure of 3BW.} The 3BW tiling is periodic, where the basic unit consists of two brick wall layers, placed one on top of the other in the way presented in Figure~\ref{fig:3BW}. Each brick wall layer is a fattening of a brick wall tiling of the plane. The underlying plane tiling is a periodic tiling, in which the basic unit consists of four columns of adjacent rectangles, where the even columns are indented by half a brick, making the tiling look like a brick wall. In the lower layer, in odd columns, the colors 1,2 are used alternately, and in even columns, the colors 3,4 are used alternately. Furthermore, the colors in the third and fourth columns are shifted by one,
see Figure~\ref{fig:3BW}. In the upper layer columns are replaced by rows. Note that once the layers are placed, the coloring of one layer fully determines the coloring of the other.



\subparagraph*{The side lengths and the slab heights.} Denote the side lengths of the bricks by $a,2a$. Clearly, the areas of all rectangles in both layers are equal, and thus, in order to make all tiles equal-volume, we choose the height of all slabs to be the same value $z$.

\medskip To find the minimal distance between two equi-colored tiles, we consider several cases:
\begin{itemize}
	\item Two equi-colored bricks in the same layer: The minimal distance between two such tiles is $a$.
	
	\item Two equi-colored bricks at adjacent layers: Here, the minimal distance is $a/2$.
	
	\item Two equi-colored bricks two layers apart: Here, the minimal distance is $z$.
\end{itemize}
Hence, the minimal distance between two equi-colored tiles is $\min\{\frac{a}{2},z\}$, and in order to optimize it we choose $z=\frac{a}{2}$.

\subparagraph*{The lower bound on $t$.} Since $z=\frac{a}{2}$, the volume of each tile is $2a \cdot a \cdot \frac{a}{2} = a^3$. Thus, in order to make the volume of all tiles equal to $1$, we fix $a=1$.
Therefore, this tiling satisfies $t = \frac{a}{2} = 0.5$. This is significantly better than the minimum distance between equi-colored tiles (i.e., the value of $t$) of the dimension-reducing tiling presented in Section~\ref{ssec:dim-reduce} (i.e.,~$t=0.261$), but is still far from the upper bound obtained in Section~\ref{ssec:BM} (namely, $t \leq 0.729$).

\subsection{Tilings based on close sphere packing}
\label{sec:sub:csp}

The tilings presented in this subsection are intended for a large number of colors. We first present the tiling in the case of $\mathbb{R}^2$, where it is easier to describe and analyze,
and then we generalize it to higher dimensions.

\subsubsection{Honeycomb of rectangles}
\label{ssec:hcr}

In the honeycomb of rectangles tiling of the plane with $k=m^2$ colors, each tile is a rectangle with side lengths $a$ and $1/a$, where
\begin{equation}\label{Eq:Honey_rect1}
a=\left(\frac{m^2-2m+1}{\frac{3}{4}m^2-m} \right)^{1/4} \geq \left(\frac{4}{3} - \frac{8}{3m}\right)^{1/4}.
\end{equation}
(The rationale behind the choice of $a$ is explained below). The tiling is periodic, where the basic unit is composed as follows. First, we construct a basic `large rectangle', which is an $m$-by-$m$ square block of tiles, using all the $k=m^2$ colors (in arbitrary order). Then, the basic unit of the tiling is two `fat rows' of adjacent large rectangles, where the second row is indented by half a large rectangle. The coloring of each large rectangle is the same. The tiling, for $k=16$, is demonstrated in Figure~\ref{fig:honey_rect}.
Note that the tiles colored in some single color resemble the shape of a honeycomb lattice (including the centers of the hexagons). This is why we call the tiling `honeycomb of rectangles'.

\begin{figure}
	\centering
	\includegraphics[width=.45\columnwidth]{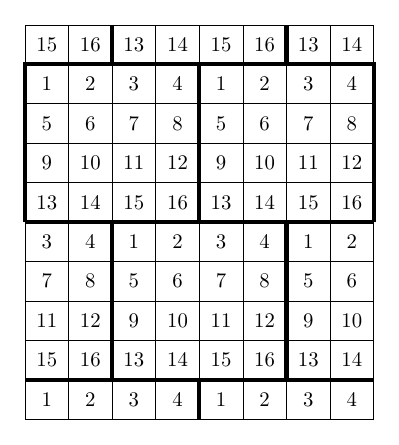}				\caption{The honeycomb of rectangles tiling for $d=2$ and $k=16$ colors. The boundaries of the basic `large rectangles' are depicted in bold. The placement of the tiles in each single color resembles the honeycomb lattice.}
	\label{fig:honey_rect}
\end{figure}

It is easy to see that the minimal horizontal and diagonal distances between two equi-colored tiles are
\[
(m-1)a \qquad \mbox{and} \qquad \sqrt{\left(\frac{m-1}{a}\right)^2+\left(\left(\frac{m}{2}-1\right)a\right)^2},
\]
respectively. By choosing $a$ such that the two distances are equal, we obtain~(\ref{Eq:Honey_rect1}), and the asymptotic lower bound
\[
t\geq \left(\frac{4}{3} - \frac{8}{3m}\right)^{1/4} \cdot (m-1) \geq \left(\frac{4}{3}\right)^{1/4} \cdot \sqrt{k}- O(1) \approx 1.074 \sqrt{k}-O(1),
\]
that matches the upper bound obtained in Section~\ref{ssec:SP-high} up to an additive $O(1)$ term. In particular, this proves the lower bound assertion of Theorem~\ref{thm:main-asymptotic} for $d=2$.

\subsubsection{Close packing of boxes}
\label{ssec:CPB}

\subparagraph*{Motivation.} This construction, a $k$-colored tiling of $\mathbb{R}^3$ where $k \gg 3$, is a natural generalization to $\mathbb{R}^3$ of the `honeycomb of rectangles' tiling presented in Section~\ref{ssec:hcr}. The idea behind the construction is to choose an optimal sphere packing in $\mathbb{R}^3$, and construct a fattened plane tiling, in which the tiles in each color are placed at the centers of the spheres of the packing. (Recall that as the number of colors is large, the size of each tile is negligible with respect to the size of its inflation, and hence, we can treat the tiles as single points.)

We use the classical HCP lattice (one of the most common \emph{close packings}, see~\cite{CS13}), that corresponds to a periodic sphere packing, in which the basic unit is two hexagonal layers of spheres, where in the top layer, each sphere is placed on top in the hollow between three spheres in the bottom layer. The coordinates of the centers of these spheres are:
\[
(r,r,r), (3r,r,r), (5r,r,r),\ldots,(2r,r+\sqrt{3}r,r),(4r,r+\sqrt{3}r,r), (6r,r+\sqrt{3}r,r), \ldots
\]
for the bottom layer, and
\small{
\[
(2r,r+\frac{\sqrt{3}}{3}r,r+\frac{2\sqrt{6}}{3}r), (4r,r+\frac{\sqrt{3}}{3}r,r+\frac{2\sqrt{6}}{3}r), \ldots,
(r,r+\frac{4\sqrt{3}}{3}r,r+\frac{2\sqrt{6}}{3}r),
(3r,r+\frac{4\sqrt{3}}{3}r,r+\frac{2\sqrt{6}}{3}r), \ldots
\]
}
for the top layer.

\subparagraph*{The structure of the tiling.} Assume that the number of colors is $k=m^3$. Each tile is a box with side lengths $(a,b,c)$ to be determined below, and the basic unit is a `large box', that is, an $m\times m \times m$ cubic block of tiles, using all the $k=m^3$ colors (in arbitrary order). Then, the basic unit of the tiling is a two-layer fattened plane tiling, in which each layer is a fattened copy of the `honeycomb of rectangles' tiling. The upper layer is shifted by $\frac{m}{2}a$ in the $x$-coordinate and by $\frac{\sqrt{3}m}{6}a$ in the $y$-coordinate, so that the corners of the large boxes lie in the coordinates of the sphere centers described above (for $r=\frac{m}{2}a$). A quick calculation shows that in order to make this possible, the proportion $(a:b:c)$ should be \emph{approximately} $(2:\sqrt{3}:\frac{2\sqrt{6}}{3})$ (where we neglect the size of each tile with respect to the size of the inflation, that can be absorbed in an $1-o(1)$ multiplicative factor in the final value of $t$). The volume of each tile is clearly $abc$. In order to make the volumes of all tiles equal to $1$, we need
\[
a \cdot \frac{\sqrt{3}}{2}a \cdot \frac{\sqrt{6}}{3}a = 1,
\]
and thus, $a=(\frac{6}{\sqrt{18}})^{1/3} = 2^{1/6} \approx 1.122$. Hence, the side lengths of each tile are
\[
(2^{1/6}, \frac{3^{1/2}}{2^{5/6}}, \frac{2^{2/3}}{3^{1/2}}) \approx (1.122, 0.972, 0.916),
\]
and the minimal distance between two equi-colored points in different tiles is
\[
(m-1)a = (2^{1/6}-o(1))k^{1/d},
\]
which matches the upper bound proved in Section~\ref{ssec:BM}. This proves the lower bound assertion of Theorem~\ref{thm:main-asymptotic} for $d=3$.

\subparagraph*{Generalization to higher dimensions.} A similar tiling can be constructed for any $d>3$ to match any \emph{lattice sphere packing}, assuming the number of colors $k$ is sufficiently large with respect to $d$. Hence, any dense lattice sphere packing can be used to obtain a lower bound on the asymptotic value of $t$ (as $k \to \infty$) in the corresponding dimension. In particular, as the \emph{$E_8$ lattice} and the \emph{Leech lattice} that attain the maximal possible density of sphere packings in dimension 8 and 24 (respectively) are lattice packings, they can be used to construct box tilings showing that the asymptotic upper bounds on $t$ in $\mathbb{R}^8$ and $\mathbb{R}^{24}$ proved in Section~\ref{ssec:BM} are tight. This proves the lower bound assertion of Theorem~\ref{thm:main-asymptotic} for $d=8,24$.

\subsection{The dimension-reducing tiling}
\label{ssec:dim-reduce}

This example shows the existence of a $(d+1)$-colored tiling of $\mathbb{R}^d$ with a positive error resilience, thus proving the `existence' part of Theorem~\ref{thm:main-qualitative}. We exemplify the tiling in $\mathbb{R}^3$, and then explain how to generalize it to higher dimensions.

\subparagraph*{Informal description of the tiling.} Informally, the tiling is constructed as follows.

 \noindent \textbf{Step~1:} We begin with dividing $\mathbb{R}^3$ into basic cubes with side length $a$ (to be determined below).
 Then, we make the `middle part' of each basic cube into a tile and give all these tiles the color~$1$. In order to keep a minimal distance of $t$ between two points colored~1 in different tiles, we must leave a neighborhood of width $t/2$ in each side of each facet of the basic cube. Hence, we are left with `fattened' walls of total width $t$. 

 \noindent \textbf{Step~2:} We make the `middle part' of each wall (i.e., fattened facet) into a tile and give all these tiles the color~$2$.
 (Note that each tile contains points from two adjacent basic cubes). In order to keep a minimal distance of $t$ between two points colored~2 in different tiles, we must leave a neighborhood of $\frac{t}{\sqrt{2}}$ near each edge (i.e., intersection of facets).
 Hence, we are left with a `fattened skeleton'.

 \noindent \textbf{Step~3:} We make the `middle part' of each edge of the skeleton into a tile and give all these tiles the color~$3$.
 (Note that each tile contains points from four adjacent basic cubes.) In order to keep a minimal distance of $t$ between two points colored~3 in different tiles, we must leave an additional neighborhood of $t/\sqrt{2}$ near each vertex (i.e., intersection of edges).
 Hence, we are left with neighborhoods of the corners (a.k.a.~vertices).

 \noindent \textbf{Step~4:} We make the neighborhoods of the corners into tiles and give them the color~$4$. (Note that each tile contains points from 8 adjacent basic cubes.) We have to make sure that the distance between each two such `fattened corners' is at least $t$, and this requirement dictates the choice of $t$. 

We call this tiling \emph{dimension-reducing} since after removing the tiles colored $1$, we are left with fattened versions of the facets (i.e., $2$-dimensional faces) of the basic cube, after removing the tiles colored $2$ we are left with fattened versions of the edges (i.e., $1$-dimensional faces), and after removing the tiles colored $3$ we are left with fattened versions of the corners (i.e., $0$-dimensional faces).




\subparagraph*{Formal definition of the tiling.} For the sake of formality, we give the exact definitions of the tiles below.

\medskip \noindent For $(x,y,z) \in \mathbb{R}^3$, let $f(x,y,z) \in [0,\frac{a}{2}]^3$ be a monotone non-decreasing ordering of $\{\bar{x},\bar{y},\bar{z}\}$, where $\bar{b}=\min\{|b-an|:n \in \mathbb{Z}\}$. That is, we measure the minimal distance to a wall in each coordinate and arrange these minimal distances in a non-decreasing order. For example, for any $a > 0$, the point $(2.4a,7.8a,3.3a)$ is mapped by $f$ to $(0.2a,0.3a,0.4a)$.
\begin{itemize}
	\item The tiles colored~1 consist of all points $(x,y,z)$ in a basic cube such that $f(x,y,z)_1 > \frac{t}{2}$. (These are exactly the points that are at least $\frac{t}{2}$-far from each wall -- the `middle part' of the cube.)
	
	\item The tiles colored~2 consist of points $(x,y,z)$ in a basic cube such that $f(x,y,z)_1 \leq \frac{t}{2}$ and $f(x,y,z)_2 > \frac{t}{2}+\frac{t}{\sqrt{2}}$. (These are the points that are close to a wall only in a single coordinate -- the `middle parts' of the fattened facets). Note that each basic cube intersects six such tiles, and each such tile contains points of two adjacent basic cubes.
	
	\item The tiles colored~3 consist of points $(x,y,z)$ in a basic cube such that $f(x,y,z)_1 \leq \frac{t}{2}$, $f(x,y,z)_2 \leq \frac{t}{2}+\frac{t}{\sqrt{2}}$, and $f(x,y,z)_3 > \frac{t}{2}+\frac{2t}{\sqrt{2}}$. (These are the points that are close to a wall in two coordinates -- the `middle parts' of the fattened edges). Note that each basic cube intersects 12 such tiles (one for each edge), and each such tile contains points of four adjacent basic cubes.
	
	\item The tiles colored~4 consist of the rest of the points (i.e., the fattened neighborhoods of the vertices). Note that each basic cube intersects 8 such tiles (one for each vertex), and each such tile contains points of eight adjacent basic cubes.
\end{itemize}

\subparagraph*{The choice of $t$.} The tiling makes sure that for $i=1,2,3$, the distance between two points colored~$i$ in different tiles is at least $t$. In order to guarantee the same condition for the color~4 as well, we have to choose $t$ such that the distance between two `neighborhoods of corners' will be at least $t$. By the construction, this amounts to the inequality
\[
a-2 \left(\frac{t}{2}+\frac{2t}{\sqrt{2}} \right) \geq t,
\]
or equivalently, $t \leq \frac{a}{2+2\sqrt{2}}$. 

\subparagraph*{The choice of $a$.} An easy computation shows that the largest tiles are those colored~1. Hence, in order to make the volume of all tiles $\leq 1$, we have to choose $a$ such that the volume of each such tile is~1. By construction, this amounts to the equality
\[
1 = (a-t)^3= a^3 \cdot \left(1- \frac{1}{2+2\sqrt{2}}\right)^3,
\]
or equivalently, $a=1+\frac{1}{1+2\sqrt{2}} \approx 1.261$. 

\subparagraph*{The lower bound on $t$.} Summarizing the above, the value of $t$ we obtain is
\[
t = \frac{a}{2+2\sqrt{2}} = \frac{1}{1+2\sqrt{2}} \approx 0.261.
\]
This lower bound for 4-colored tilings of  $\mathbb{R}^3$ is superseded by the `3-dimensional brick wall' tiling presented in Section~\ref{ssec:bricks}, which obtains $t=0.5$. 

\subparagraph*{Generalization to $d>3$.} The construction defined above generalizes naturally to tilings of $\mathbb{R}^d$ with tiles of $d+1$ colors. For $1 \leq i \leq d$, the tiles colored~$i$ consist of points $(x_1,\ldots,x_d)$ in a basic cube such that $f(x_1,\ldots,x_d)_j \leq \frac{t}{2}+\frac{(j-1)t}{\sqrt{2}}$ for all $j<i$, and $f(x_1,\ldots,x_d)_i > \frac{t}{2}+\frac{(j-1)t}{\sqrt{2}}$. (These are the points that are close to a wall in $i-1$ coordinates -- the `middle parts' of the $(d-i+1)$-dimensional faces). The tiles colored~$d+1$ consist of the rest of the points (i.e., the fattened neighborhoods of the vertices).

By the same computation as above, in order to guarantee a distance of $\geq t$ between two equi-colored tiles, we have to choose $$t \leq \frac{a}{2+2(d-1)\sqrt{2}}.$$ In order to make the volume of all tiles $\leq 1$, we have to take $$a \leq 1+\frac{1}{1+2(d-1)\sqrt{2}}.$$ By choosing $a$ to be the largest possible value, we obtain a tiling with minimum distance of $$t=\frac{1}{1+2(d-1)\sqrt{2}} = \Omega \left(\frac{1}{d} \right)$$ between equi-colored tiles. This lower bound of $t \geq \Omega(\frac{1}{d})$ for $(d+1)$-colored tilings of $\mathbb{R}^d$ is not very far from the upper bound we obtained in Section~\ref{ssec:BM} -- namely, $O(\frac{\log d}{\sqrt{d}})$.

\subparagraph*{Related work.} A variant of the dimension-reducing tiling was presented in~\cite[Theorem~5.3]{CzumajJK0Y22} and used in the approximation algorithm for the Euclidean uniform facility location problem presented there. In the construction of~\cite{CzumajJK0Y22}, the tiles are built in a dimension-increasing manner -- the first group of tiles consists of neighborhoods of vertices, the second group of tiles consists of neighborhoods of edges (i.e., $1$-dimensional faces), etc. As is stated in~\cite{CzumajJK0Y22}, the order of magnitude of $t$ obtained by their construction is the same as the order of magnitude obtained by our dimension-reducing tiling. (Checking this requires some computation, since the normalization used in~\cite{CzumajJK0Y22} is bounding the diameter of each tile, while we normalize by bounding the volume of each tile).


\section{Open problems} 

While we fully solved the qualitative question of minimizing the number of colors required for achieving positive error resilience, several questions remain open regarding the maximal resilience rate that can be achieved for a given number of colors. In particular, for dimensions $2,3,8,24$ we determined the exact asymptotic resilience for a large number $k$ of colors, using a connection to the \emph{densest sphere packing} problem. When only very few colors are allowed, the situation is much less clear. For example, we do not even know whether the brick wall constructions we present in Section~\ref{ssec:bricks} have the highest error resilience among space partitions in $\mathbb{R}^2$ with $3$ and $4$ colors. It will be interesting to obtain new upper bounds via different techniques or new lower bound constructions.

\section*{Acknowledgements}

We are grateful to Stephen D. Miller for inspiring discussions on sphere packing, and to anonymous reviewers for referring us to results on the Assouad-Nagata dimension and to the line of works on sparse partitions.

\appendix

\section{Alternative Normalization and Distance Function}
\label{App:Normalization}

\subparagraph*{Normalization.} The primary reason behind our choice to normalize the tiles by upper-bounding their \emph{volume} is that this normalization complies with using the Brunn-Minkowski theorem and with the reduction to sphere packing -- that are our main tools for proving upper bounds.
On the qualitative level, our results can be easily translated to results with respect to other natural normalizations. For example, the assumption that \emph{the radius of each tile is at most $1$} (which means that the $L_2$ distance between each $x \in \mathbb{R}^d$ and its rounded value is at most $1$) implies that the volume of each tile is at most the volume of the unit ball in $\mathbb{R}^d$, which allows translating all our error resilience upper bounds to this `bounded radius' setting. As for the lower bounds, they come from explicit constructions whose error resilience can be recomputed with respect to any other natural normalization.

It is not clear however that our lower bound constructions are quantitatively optimal with respect to other normalizations.

\subparagraph*{Alternative metrics.} Instead of the $L_2$ distance (which is probably the most natural distance metric), one may consider the error resilience problem with respect to other distance metrics.

For example, in order to measure the error resilience with respect to the $L_\infty$ distance, one has to inflate each tile $T$ into $T''=\{y: \exists x \in T, \max_i |x_i-y_i|<r\}$. It turns out that the problem is much easier with respect to this metric. Indeed, assume that the number of colors is $k=m^d$ for some $2 \leq m \in \mathbb{N}$. Consider a periodic tiling, in which each basic unit is a cube with side length $m$ that consists of $m^d$ unit-cube tiles, each colored in a different color. This tiling achieves error resilience of $\frac{m-1}{2}$. On the other hand, it is easy to see that the argument via the Brunn-Minkowski theorem presented in Section~\ref{ssec:BM} implies that any tiling of $\mathbb{R}^d$ with $k=m^d$ colors achieves error resilience of at most $\frac{m-1}{2}$ with respect to the $L_{\infty}$ metric, and thus, the cubic tiling we described is optimal. Lower and upper bounds for other values of $k$ can be obtained by variants of this tiling and the corresponding upper bound proof.

The difference between the metrics comes from the fact that in `cubic' inflation (which is done with respect to the $L_{\infty}$ metric), non-intersecting inflations of cubic tiles fill the entire space, while in inflation by balls (that we have with respect to the $L_2$ metric), large gaps are left between the inflated tiles.

\end{document}